\documentclass{article}
 
\usepackage{amsfonts}
\usepackage{amssymb}

\usepackage{graphicx,framed}
\usepackage{amsmath,mathrsfs}
\usepackage{epstopdf}
\usepackage{epsfig}

\newtheorem{theorem}{Theorem}

\newtheorem{proposition}[theorem]{Proposition}
\newtheorem{remark}[theorem]{Remark}

\newenvironment{proof}[1][Proof]{\textbf{#1.} }{\ \rule{0.5em}{0.5em}}

\begin{document}

\title{How to estimate past measurement interventions on a quantum system undergoing continuous monitoring}
\author{John E. Gough \\
Aberystwyth University,\\
Ceredigion, SY23 3BZ, Wales, UK}

\date{}

\maketitle

\begin{abstract}
We analyze the problem of estimating past quantum states of a monitored system from a mathematical perspective in order to ensure self-consistency with the principle of quantum non-demolition. Despite several claims of ``measuring noncommuting observables'' in the physics literature, we show that we are always measuring commuting processes. Our main interest is in the notion of quantum smoothing or retrodiction. In particular, we examine proposals to estimate the result of an external measurement made on an open quantum systems during a period where it is also undergoing continuous monitoring. A full analysis shows that the non-demolition principle is not actually violated, and so a well-posed as a statistical inference problem can be formulated. We extend the formalism to consider multiple independent external measurements made on the system over the course of a continual period of monitoring.
\end{abstract}

\section{Introduction}
Statistical inference uses Bayes Theorem to calculate estimates for unmeasured variables in terms of measured ones. However, this may be applied to quantum theory \textit{only} in the case where all these quantities are described by compatible observables. It is legitimate to estimate the value of a present (filtering) or a future (prediction) observable of an open system based on past measurement of the output processes as these all commuting. But past observables do not commute with future outputs so their estimation (quantum smoothing) violates the non-demolition principle and constitutes a misapplication of Bayes Theorem. 

Nevertheless, there have been several proposals for quantum smoothing \cite{Tsang_A_09}-\cite{LCW_19} addressing the question of whether useful information about the past may still be extracted from future measurements. Here we revisit the proposal by 
Gammelmark, Julsgaard, and K. M\o lmer \cite{Gammelmark13} who consider the situation of an system continuously monitored over the time interval $[0,T]$ upon which we make an external instantaneous measurement at some intermediate time $\tau$. This has been applied to experimental schemes monitoring the fluorescence of superconducting qubits \cite{C14}.

The external measurement is taken to be indirect - that is, we follow von Neumann's scheme of preparing a quantum probe in a state, coupling it to the system at time $\tau$ and then making an instantaneous measurement of a probe observable. It is assumed that the probe is initially not entangled with either the system or it's environment, and that it only couples to the system itself - and not directly with the environment. We refer to such measurements as external \textit{measurement interventions} on the open system. Clearly such interventions perturb the dynamics so that not only is the system affected, but the intervention will also be apparent in the continuous measurement readout. The outcome of the intervention measurement may be concealed, so one may ask for the probability of its outcome conditional on the entire continuously monitored output - both past and future.

The central question is whether such a scheme satisfies the non-demolition principle. If it does not, then we have the situation that the measurements made beyond time $\tau$ degrade the estimate as we are now measuring incompatible observables, and theoretically we have no ground to be applying Bayes Theorem. Following a careful examination of what commutes with what, our main result is to show that the scheme does actually satisfy the non-demolition principle and to use standard quantum filtering theory to provided the least squares estimate for the probe's observable. This rederives the results of 
Gammelmark, \textit{et al.} \cite{Gammelmark13}. In section \ref{sec:MultiTime}, we extend the situation to several independent measurement interventions made over the course of the time interval $[0,T]$.

\subsection{Bayes Rule}

In classical probability we may assign a joint probability $P\left[ A\cap B%
\right] $ for events $A$ and $B$ to occur. (This is always possible because
the events are subsets of the sample space $\Omega $ and so their
intersection always exists as a set.) The conditional probability for $A$ to
occur given $B$ is $P\left( A|B\right) =P\left( A\cap B\right) /P\left(
B\right) $. We then have $P\left( A\cap B\right) =P\left( A|B\right) P\left(
B\right) $ and by symmetry this equals $P\left( B|A\right) P\left( A\right) $
from which we obtain
\begin{eqnarray}
\text{Bayes Theorem: \ \ }P\left( A|B\right) =\frac{P(B|A)P(A)}{P\left(
B\right) }.
\label{eq:Bayes}
\end{eqnarray}
In practice, this is used as follows: we have a model which tells us the probability (technically referred to as \textit{likelihood}) of $B$ occurring if $A$ occurs; it may be difficult/infeasible to determine if $A$ occurs in an experiment, but easy to determine $B$; we then use Bayes Theorem compute the probability of $A$ to occur conditional on seeing that $B$ occurs in experiment. (The usual caveat is that we often do not know $P(A)$, in which case we have to make a guess known as the \textit{prior}.)

The key question is \emph{under what conditions may we apply Bayes Theorem in quantum theory?}

In quantum theory we may encounter probabilities $P\left( A;B\right) $
meaning that we observe event $A$ first then event $B$. But these are not
joint probabilities, they are sequential probabilities and in general we
have $P\left( A;B\right) \neq P\left( B;A\right) $. For instance, suppose
the state is given by density matrix $\rho $ and $A$ and $B$ correspond to
projections $Q$ and $P$ respectively, then
\begin{eqnarray*}
P\left( A;B\right)  &=&\text{tr}\left\{ P\rho PQ\right\} =\text{tr}\left\{
\rho PQP\right\} , \\
P\left( B;A\right)  &=&\text{tr}\left\{ Q\rho QP\right\} =\text{tr}\left\{
\rho QPQ\right\} .
\end{eqnarray*}
If the events are compatible then the projections commute, and so $P\left(
A;B\right) =P\left( B;A\right) =\mathrm{tr}\left\{ \rho PQ\right\} =\mathrm{%
tr}\left\{ \rho QP\right\} $ unambiguously for any state. More generally,
however, if $A$ and $B$ are incompatible then we have no statistical
foundation for using Bayes Theorem. In other words, \emph{while you can numerically use Bayes Theorem for quantum problems, the answers you obtain may be wholly meaningless!}

\bigskip 

A simple description of indirect measurement is the following. Suppose our system has Hilbert space $\mathfrak{h}$ and in prepared in a state with density matrix $\rho$. We set up a probe which has Hilbert space $\mathfrak{H}$ and which is independently prepared in a state $\rho_{\text{probe}}$ and couple it to the system using some unitary $V$. We measure an observable $M$ of the probe which we take to have spectral decomposition $ M= \sum_m P_m$. If we measure the eigenvalue $m$ for the probe observable $M$ then the state of the system updates to
\begin{eqnarray}
\rho \mapsto \frac{1}{ \text{tr} \big\{ \Phi_m (\rho )\big\} }\, \Phi_m ( \rho  ) .
\end{eqnarray}
where $\Phi_m$ is the CP map
\begin{eqnarray}
\Phi_m ( \rho ) = \text{tr}_{\mathfrak{H}} \big\{ V \rho \otimes \rho_{\text{probe}} V^\ast \, I \otimes P_m \big\} .
\end{eqnarray}

As a special case, let $ \{ | m \rangle \}$ form an orthonormal basis for the probe Hilbert space and take the probe to be in the pure state $| \varphi \rangle$. Suppose that, for each $\vert \psi \rangle $  in $\mathfrak{h}$, we have $ V \, |\psi \rangle \otimes \vert \varphi  \rangle = \sum_m V_m |\psi \rangle \otimes | m \rangle$ where the $V_m$'s are operators on the system space. Then
$\Phi_m (\rho ) \equiv V_m \rho V_m^\ast$.

\begin{eqnarray}
\rho \mapsto \frac{1}{ \text{tr} \big\{ \rho V^\ast_m V_m \big\} }\, V_m \rho  V_m^\ast 
\end{eqnarray}

\subsection{Non-Demolition Principle}
There are two simple rules that need to be followed when trying to extract
quantum information from measurements.

\begin{framed}
\noindent \textbf{Rule \# 1 (Self-Non-Demolition Principle)}:\\
 All the observables measured
should be compatible with each other.

\bigskip

\noindent \textbf{Rule \# 2 (Non-Demolition Principle)}:\\
All the observables estimated should be
compatible with all measured observables.
\end{framed}

\bigskip

If we wish to infer useful information about an observable $X$ from measured
observables $\left\{ Y_{1},\cdots ,Y_{N}\right\} $\ and, then the principles
combined will require that the family $\left\{ X,Y_{1},\cdots ,Y_{N}\right\} 
$ is commutative. In this case there exist well-defined joint probabilities
and Bayes Theorem may be applied to condition $X$ on the observations. If
the rules are not followed then we are misapplying Bayes Theorem and,
statistically speaking, are in a state of sin.

In general, given a state $\mathbb{E} = \text{tr} \{ \rho \cdot \}$ we may define the conditional expectation of $X$ given $Y$ provided they are compatible. 
Let $P_X [dx] $ and $P_Y [dy]$ be the projection-valued measures associate with $X$ and $Y$ respectively. In this case $p_{X,Y} [dx , dy ] = \text{tr} \{ \rho \, P_X [dx] P_Y [dy] \}$ is a well-defined joint probability measure for both observables and the conditional probability $p_{X|Y} [ dx |y ]$ is well-defined via
$p_{X|Y} [ dx |y ] p _Y [ dy ] \equiv p_{X,Y} [dx,dy]$ we set 
\begin{eqnarray}
\mathbb{E} [ X|Y ] = \int x \, \hat{P}_X [dx]
\end{eqnarray}
where
\begin{eqnarray}
\hat{P}_X [dx]
= \int p_{X|Y} [ dx |y ] P _Y [ dy ] .
\end{eqnarray}
The $\hat P_{X|Y} [dx]$ is an effect valued measure.

In principle, for a given state $\mathbb{E}$, we can always define the conditional expectation $\mathbb{E} [ X | \mathfrak{Y} ]$ of an observables $X$ onto the algebra $\mathfrak{Y}$ generated by the measured observables provided that the non-demolition rules above hold.

\subsection{Quantum Trajectories (Filtering)}

Let us consider a fairly standard set up in quantum trajectories. We have a
system with state space $\mathfrak{h}$ coupled to a bath $\mathfrak{F}$. The bath
consists of $n$ input processes $b_{\text{in},k}\left( t\right) $ satisfying
singular commutation relations $\left[ b_{\text{in},j}\left( t\right) ,b_{%
\text{in},k}\left( s\right) ^{\ast }\right] =\delta _{jk}\delta \left(
t-s\right) $ and we assume the vacuum state $|\Omega \rangle $ for the bath.
The coupling of the bath to the system over the time interval $t_{1}$ to $%
t_{2}$ is taken to be described by the unitary on $\mathfrak{h}\otimes \mathfrak{F}$%
\begin{eqnarray*}
U\left( t_{2},t_{1}\right) =\vec{T}e^{-i\int_{t_{1}}^{t_{2}}\Upsilon \left(
s\right) ds}
\end{eqnarray*}
where $-i\Upsilon \left( t\right) =\sum_{k}\left\{ L_{k}\otimes b_{\text{in}%
,k}\left( t\right) ^{\ast }-L_{k}^{\ast }\otimes b_{\text{in},k}\left(
t\right) ^{\ast }\right\} -iH$. We note the flow property
\begin{eqnarray}
U\left( t_{3},t_{2}\right) U\left( t_{2},t_{1}\right) =U\left(
t_{3},t_{1}\right) ,\quad \left( t_{3}\geq t_{2}\geq t_{1}\right) .
\label{eq:flow}
\end{eqnarray}

We intend to measure the field quadratures. To this end, set $Z_{k}\left(
t\right) =\int_{0}^{t}\left\{ b_{\text{in},k}\left( s\right) +b_{\text{in}%
,k}\left( s\right) ^{\ast }\right\} ds$, then we note that
\begin{eqnarray}
\left[ Z_{j}\left( t\right) ,Z_{k}\left( s\right) \right] =0.
\label{eq:Z_commutes}
\end{eqnarray}
However, what we must measure are the outputs
\begin{eqnarray*}
Y_{k}\left( t\right) = U\left( t,0\right) ^{\ast }\left\{ I\otimes
Z_{k}\left( t\right) \right\} U\left( t,0\right)  =\int_{0}^{t}\left\{ b_{\text{out},k}\left( s\right) +b_{\text{out}%
,k}\left( s\right) ^{\ast }\right\} ds.
\end{eqnarray*}
One may show that
\begin{eqnarray*}
b_{\text{out},k}\left( t\right) =I\otimes b_{\text{in},k}\left( t\right)
+j_{t}\left( L_{k}\right) 
\end{eqnarray*}
where, for any system operator $X$,  
\begin{eqnarray*}
j_{t}\left( X\right) =U\left( t,0\right) ^{\ast }\left\{ X\otimes I\right\}
U\left( t,0\right) .
\end{eqnarray*}

\begin{proposition}
\label{prop:self_ND_basic}
The family $\big\{ Y_{k}\left( t\right) :k=1,\cdots ,n,0\leq t\leq
T\big\} $ is a commutative family.
\end{proposition}

\begin{proof}
Suppose that $t\geq s$, then we observe that
\begin{eqnarray}
Y_{k}\left( s\right) =U\left( t,0\right) ^{\ast }\left\{ I\otimes
Z_{k}\left( s\right) \right\} U\left( t,0\right) .
\label{eq:Y_extend}
\end{eqnarray}
To see this, we use (\ref{eq:flow}) to write the right hand side as
\begin{eqnarray*}
U\left( s,0\right) ^{\ast }U\left( t,s\right) ^{\ast }\left\{ I\otimes
Z_{k}\left( s\right) \right\} U\left( t,s\right) U\left( s,0\right) 
\end{eqnarray*}
however $U\left( t,s\right) $ couples the system to only those input
processes over the time interval $s$ to $t$ and so commutes with $I\otimes
Z_{k}\left( t\right) $, and by unitarity $U\left( t,s\right) ^{\ast }U\left(
t,s\right) $ is the identity. We therefore have that
\begin{eqnarray*}
\left[ Y_{j}\left( t\right) ,Y_{k}\left( s\right) \right] =U\left(
t,0\right) ^{\ast }\left\{ I\otimes \left[ Z_{j}\left( t\right) ,Z_{k}\left(
s\right) \right] \right\} U\left( t,0\right) 
\end{eqnarray*}
which vanishes by (\ref{eq:Z_commutes}).
\end{proof}

\begin{proposition}
\label{prop:NDP_basic}
The non-demolition principle holds for $j_{t}\left( X\right) $ and the
measured observables $\big\{ Y_{k}\left( u\right) :k=1,\cdots ,n,0\leq
u\leq s\big\} $ provided $t\geq s$.
\end{proposition}

\begin{proof}
For $t\geq s$ we may use (\ref{eq:Y_extend}) again to show that
\begin{eqnarray*}
\left[ j_{t}\left( X\right) ,Y_{j}\left( s\right) \right]  =U\left(
t,0\right) ^{\ast }\left[ X\otimes I,I\otimes Z_{k}\left( s\right) \right]
U\left( t,0\right)  =0.
\end{eqnarray*}
\end{proof}

We may therefore estimate the present (filtered), or future (predicted),
value of a system observable based on the quadrature observations up to the
present time. Estimating past values (smoothing) is not possible as it
violates the non-demolition principle.

\bigskip 

The filter is well-known for this problem. The best estimate for $j_{t}\left( X\right) $ is 
the conditional expectation
\begin{eqnarray}
\pi_t (X) = \mathbb{E} [ j_t (X) | \mathfrak{Y}_{[0,t]} ],
\end{eqnarray}
where $\mathfrak{Y}_{[0,t]}$ is the algebra generated by $\big\{ Y_{k}\left( u\right) :k=1,\cdots ,n,0\leq
u\leq t\big\} $.

We may write $\pi _{t}\left( X\right) =tr\left\{ \hat{\rho}%
_{t}X\right\} $ where $\hat{\rho}_{t}$ satisfies the stochastic partial
differential equation (SPDE)
\begin{eqnarray}
d\hat{\rho}_{t} =i\left[ H,\hat{\rho}_{t}\right] dt+ \mathcal{D} ( \hat \rho_t ) \, dt + \sum_{k}\left( L_{k}\hat{\rho}_{t}+\hat{\rho}_{t}L_{k}^{\ast }-\lambda
_{k}\left( t\right) \hat{\rho}_{k}\right) \,dI_{k}\left( t\right) 
\label{eq:filter_basic}
\end{eqnarray}
where $\mathcal{D} ( \hat \rho_t ) =\sum_{k}\left( L_{k}%
\hat{\rho}_{t}L_{k}^{\ast }-\frac{1}{2}L_{k}^{\ast }L_{k}\hat{\rho}_{t}-%
\frac{1}{2}\hat{\rho}_{t}L_{k}^{\ast }L_{k}\right) dt$, $\lambda _{k}\left( t\right) =tr\left\{ \hat{\rho}_{t}\left(
L_{k}+L_{k}^{\ast }\right) \right\} $ and the $I_{k}\left( t\right) $ are
Wiener processes known as the \emph{innovations processes} and are given by
\begin{eqnarray}
dI_{k}\left( t\right) =dY_{k}\left( t\right) -\sqrt{\eta _{k}}\lambda
_{k}\left( t\right) dt,
\label{eq:innovations}
\end{eqnarray}
where we allow for an efficiency factor $\eta _{k}\in [ 0,1]$ for the $k$ th measurement. The derivation of the filter makes explicit use of Bayes Theorem.

We remark that the SPDE (\ref{eq:filter_basic}) is nonlinear in $\hat{\rho}_t$, however,it is advantageous to work with an unnormalized version which obeys a linear SDPE. It is possible to write $\hat{\rho}_t = \rho_t / \text{tr} \{ \rho+t \} $ where $\rho_t$ satisfies 
\begin{eqnarray}
d \rho_{t} =i\left[ H, \rho_{t}\right] dt+ \mathcal{D} (   \rho_t ) \, dt   +\sum_{k}\left( L_{k} \rho_{t}+ \rho_{t}L_{k}^{\ast } \right) \,dY_{k}\left( t\right) 
\label{eq:filter_zakai}
\end{eqnarray}

\begin{remark}
The filter $\hat \rho _t$ is to be obtain by solving the SPDE (\ref{eq:filter_basic}) subject to an initial condition, normally $\hat \rho_0 = \rho_0$.
The initial state $\rho_0$ may not always be known in which case guessing at one may lead to error. In many cases, the sensitivity to initial conditions 
is not important and a wrongly initialized filter will converge to the correct value asymptotically in time.
\end{remark}

\begin{remark}
From the readout $Y_k$ we may extract the values $\lambda _{k}\left( t\right) =tr\left\{ \hat{\rho}_{t}\left( L_{k}+L_{k}^{\ast }\right) \right\} $. This is colloquially phrased as ``measuring the observable $ L_{k}+L_{k}^{\ast }  $.'' However, what we are doing, in fact, is measuring the field quadrature $Y_k (t) $ and estimating the expected value of $ L_{k}+L_{k}^{\ast }  $ from this.
The phrase is not harmless! The $ L_{k}+L_{k}^{\ast }  $ need not commute for different $k$ and so, while arresting, it is misleading to say we are measuring noncommuting observables. In reality, Proposition \ref{prop:self_ND_basic} shows that the quadratures $Y_k (t)$ commute for all $k$ and all $t \ge 0$. So the totality of what we measure is compatible - however, we can use this to estimate noncommuting observables.
\end{remark}

\subsection{Computing Estimates}
Let us now sketch the argument for deriving the estimates. We adapt the presentation given in \cite{BvHJ}.

\begin{theorem}[Bouten-van Handel]
Let $\mathbb{E} = \text{tr} \{ \rho \cdot \}$ be a state on an algebra $\mathfrak{A}$ of operators and $\mathfrak{Y}$ a commutative subalgebra generated by the measured observables in a given experiment. Suppose that $\mathfrak{Y} = U^\ast \mathfrak{Z} U$ for a given unitary $U$. Furthermore, suppose that there is a fixed $F \in \mathfrak{Y}^\prime$ such that, for every $X \in \mathfrak{Y}^\prime$, we have
\begin{eqnarray}
\mathbb{E} [U^\ast XU ] = \mathbb{E} [ F^\ast XF ] .
\end{eqnarray}
Then
\begin{eqnarray}
\mathbb{E} [U^\ast X U  | \mathfrak{Y} ] = \frac{1}{\sigma (I) } \, \sigma (X) ,
\end{eqnarray}
where
\begin{eqnarray}
\sigma (X) = U^\ast \mathbb{E} [ F^\ast X F | \mathfrak{Z} ] U.
\end{eqnarray}
\end{theorem}

For completeness, we sketch the main argument of the proof: wWe first note that
\begin{eqnarray}
\mathbb{E} [U^\ast X U | \mathfrak{Y} ] = U^\ast \mathbb{E}^\prime [X| \mathfrak{Z} ] U ,
\label{eq:part1}
\end{eqnarray}
where $\mathbb{E}^\prime$ is the expectation with respect to the state $\rho^\prime = U \rho U^\ast$, then, with $Z \in \mathfrak{Z}^\prime$, we see that
\begin{eqnarray*}
\mathbb{E} \big[ \mathbb{E} [ F^\ast X F | \mathfrak{Z} ] Z \big] &=& \mathbb{E} [F^\ast XF Z]
= \mathbb{E} [F^\ast X ZF] = \mathbb{E}^\prime  [ X Z]\nonumber \\
&=&  \mathbb{E}^\prime   \big[ \mathbb{E}^\prime [X | \mathfrak{Z} ]  Z \big] \nonumber \\
&=&  \mathbb{E}    \big[ F^\ast \mathbb{E}^\prime [X | \mathfrak{Z} ]  Z F \big] \nonumber \\
&=&  \mathbb{E}    \big[ F^\ast F\mathbb{E}^\prime [X | \mathfrak{Z} ]  Z \big] \nonumber \\
&=&  \mathbb{E}    \bigg[ \mathbb{E} \big[ F^\ast F \mathbb{E}^\prime [X | \mathfrak{Z} ]  Z | \mathfrak{Z} \big] \bigg] \nonumber \\
&=&  \mathbb{E}    \bigg[ \mathbb{E} \big[ F^\ast F  | \mathfrak{Z} \big] \mathbb{E}^\prime [X | \mathfrak{Z} ] Z \bigg] ;\nonumber \\
\end{eqnarray*}
and from this, we deduce that
\begin{eqnarray}
\mathbb{E}^\prime [X | \mathfrak{Z} ]
=\mathbb{E} \big[ F^\ast F  | \mathfrak{Z} \big]^{-1} 
\, \mathbb{E} \big[ F^\ast X F  | \mathfrak{Z} \big] .
\label{eq:part2}
\end{eqnarray}
Combining (\ref{eq:part1}) and (\ref{eq:part2}) gives the desired result.

To see how to apply this to the filtering problem, let us first of all note that 
the algebra $\mathfrak{Y}_{[0,t]}$ generated by the output quadratures over the time interval $[0,t]$ may be written as
$ U(t)^\ast [ I \otimes \mathfrak{Z}_{[0,t]} ] U(t)$ by virtue of (\ref{eq:Y_extend}) from Proposition \ref{prop:self_ND_basic}. 
We next observe that we may write
\begin{eqnarray*}
\mathbb{E} [ U(t,0)^\ast [ X \otimes I ] U(t,0) ] = 
\mathbb{E} [ F(t,0)^\ast [ X \otimes I ] F(t,0) ] 
\end{eqnarray*}
where we introduce the two-parameter family of time-ordered exponentials
\begin{eqnarray}
F(t_2 , t_1 ) = \vec{T} e^{\int_{t_1}^{t_2} [ \sum_k L_k dZ_k (t) + K dt ] } ,
\label{eq:F_Z}
\end{eqnarray}
where $K= - \frac{1}{2} \sum_k L^\ast_k L_k -iH$.
Here the state is a tensor product of system and Fock vacuum state of the bath: $\rho_{\text{sys}} \otimes | \Omega \rangle \langle \Omega |$. Indeed it is easy to show that
\[
U(t_2 , t_1 ) \, |\psi \rangle \otimes | \Omega \rangle = 
F(t_2 , t_1 ) \, |\psi \rangle \otimes | \Omega \rangle ,
\]
for any system vector $| \psi \rangle $. 

We use the Theorem to write
\begin{eqnarray}
\pi_t (X) =  \mathbb{E} [ j_t (X) | \mathfrak{Y}_{[0,t]} ]  
= \sigma_t (I)^{-1} \, \sigma_t (X) ,
\end{eqnarray}
with
\begin{eqnarray}
\sigma_t (X) = \text{tr} \big\{ \rho_{\text{sys}} F^Y (t,0)^\ast X F^Y (t,0) \big\} ,
\end{eqnarray}
and where we introduce the new process
\begin{eqnarray}
F^Y(t_2 , t_1 ) = \vec{T} e^{\int_{t_1}^{t_2} [ \sum_k L_k dY_k (t) + K dt ] } ,
\label{eq:F_Y}
\end{eqnarray}

It is not difficult to see that $\sigma_t (X)$ satisfies the stochastic differential equation  
\begin{eqnarray}
d \sigma_t (X) = \sigma_t ( \mathcal{L} X ) \, dt + \sum_k \sigma_t (XL_k + L^\ast_k X ) \, dY_k (t) .
\label{eq:BZ}
\end{eqnarray}
where the (Heisenberg picture) Lindbladian is 
\begin{eqnarray}
\mathcal{L} X = \sum_k L^\ast_k XL_k + X K + K^\ast X.
\label{eq:Lindblad}
\end{eqnarray}
This is equivalent to (\ref{eq:filter_basic}) for the unit efficiency case. 
In the classical world, we encounter the Zakai equation which is a \textit{linear} SPDE for the unnormalized filter which occurs in continuous time estimation problems. The equation \ref{eq:BZ} is
its quantum analogue and we refer to it as the \emph{Belavkin-Zakai equation}.

\section{Intervention Measurements}

The proposal was made by Gammelmark \textit{et alia} that, in addition to
the quadrature measurements over the time interval $\left[ 0,T\right] $, one
might also make a further measurement at an intermediate time $\tau $ $%
\left( 0<\tau <T\right) $. The question we ask here is whether we can
actually do this without violating the non-demolition principle.

\bigskip 

If, rather than an instantaneous measurement at time $\tau $, we are
prepared to settle for another indirect measurement then there is a simple
way to do this. The additional measurement is realized as a quadrature
measurement of an extra input field and so we just have the old filtering
problem with a larger number of channels.

\bigskip 

The issue, of course, is that we want to make a measurement at a fixed time $\tau $. It is clear that if this is a direct measurement of some observable $A$ of the system then what in practice we are doing is measuring $j_{\tau }\left( A\right) $ which we know commutes with the output quadratures up to time $\tau $, but not beyond. This necessitates using indirect measurements.
To this end we introduce an auxiliary system with Hilbert space $\mathfrak{H}$ which will act as the probe. For simplicity, we assume that the auxiliary system is in an initial state $\rho_{\text{probe}} $ and has no internal evolution. We now work on the tensor product $\mathfrak{h}\otimes \mathfrak{F}\otimes \mathfrak{H}$. and have the same coupling between the system and bath as
before, no coupling between that bath and probe, and no interaction between the system and probe save for a unitary $V$ which is applied at time $\tau $. We understand that $\left[ V,I\otimes F\otimes I\right] =0$ for all bath operators $F$, however, $V$ must lead to a nontrivial coupling between the system and probe so that we may obtain any information about the system by
measuring the probe.

The unitary dynamics on $\mathfrak{h}\otimes \mathfrak{F}\otimes \mathfrak{H}$ is now described by
\begin{eqnarray}
\widetilde{U}\left( t,0\right) =\left\{ 
\begin{array}{ll}
U ( t,0 ) \otimes I , & t < \tau ; \\ 
\left[ U ( t ,\tau  ) \otimes I \right]  V  [ U ( \tau , 0 ) \otimes I], & t\geq \tau .
\end{array}
\right. 
\end{eqnarray}
We also set
\begin{eqnarray*}
\widetilde{U}\left( t,s\right) =\widetilde{U}\left( t,0\right) \widetilde{U}\left(
s,0\right) ^{\ast }\quad \left( t\geq s\right) .
\end{eqnarray*}
(The two-parameter family $\widetilde{U}(t,s)$ therefore satisfies the flow property.)

The output quadratures are now
\begin{eqnarray}
\widetilde{Y}_{k}\left( t\right) =\widetilde{U}\left( t,0\right) ^{\ast }\left[
I\otimes Z_{k}\left( t\right) \otimes I\right] \widetilde{U}\left( t,0\right) .
\end{eqnarray}
Likewise, the observable of the probe which is measured will be of the form
\begin{eqnarray}
\widetilde{M}\left( \tau \right)  &=&\widetilde{U}\left( \tau ,0\right) ^{\ast }%
\left[ I\otimes I\otimes M\right] \widetilde{U}\left( \tau ,0\right)  \nonumber \\
&=&[ U(\tau , 0 ) \otimes I]^\ast V^{\ast }\left[ I\otimes I\otimes M\right] V [ U(\tau , 0 ) \otimes I] ,\nonumber\\
&& \,
\label{eq:tilde_M}
\end{eqnarray}
for some fixed self-adjoint operator $M$ on $\mathfrak{H}$. Finally, system observables will have the modified evolution
\begin{eqnarray}
\widetilde{\jmath}_{t}\left( X\right) =\widetilde{U}\left( t,0\right) ^{\ast }\left[
X\otimes I\otimes I\right] \widetilde{U}\left( t,0\right) .
\label{eq:tilde_j}
\end{eqnarray}

\begin{proposition}
\label{prop:tilde_Y_commute}
The family $\big\{ \widetilde{Y}_{k}\left( t\right) :k=1,\cdots ,n,0\leq t\leq
T\big\} $ is a commutative family.
\end{proposition}

\begin{proof}
Suppose that $t\geq s$, then similar to (\ref{eq:Y_extend}) we observe that
\begin{eqnarray}
\widetilde{Y}_{k}\left( s\right) =\widetilde{U}\left( t,0\right) ^{\ast }\left\{
I\otimes Z_{k}\left( s\right) \otimes I\right\} \widetilde{U}\left( t,0\right) .
\label{eq:tilde_Y_extend}
\end{eqnarray}
The right hand side is
\begin{eqnarray*}
\tilde U \left( t,0\right) ^{\ast } \tilde U\left( t,s\right) ^{\ast }\left\{ I\otimes
Z_{k}\left( s\right) \otimes I\right\} \tilde U \left( t,s\right) \tilde U\left( t,0\right)
.
\end{eqnarray*}
For $t\geq s\geq \tau $, the unitary $U\left( t,s\right) $ couples only the input processes over the time interval $s$ to $t$ to the system and so commutes with $I\otimes Z_{k}\left( t\right) \otimes I$. If $t\leq \tau \leq s$, then we have an extra factor $V$ to worry about, but this only couples the system and probe and so again commutes with $I\otimes Z_{k}\left( s\right) \otimes I
$. Finally, if $\tau >t\geq s$, then we are back in the situation of Proposition \ref{prop:self_ND_basic}. We therefore have that $\left[ \widetilde{Y}_{j}\left( t\right) ,\widetilde{Y}_{k}\left( s\right) \right] $ equals
\begin{eqnarray*}
\widetilde{U}\left( t,0\right) ^{\ast }\left\{ I\otimes \left[ Z_{j}\left(
t\right) ,Z_{k}\left( s\right) \right] \otimes I\right\} \widetilde{U}\left(
t,0\right)  =0.
\end{eqnarray*}
\end{proof}

\begin{proposition}
\label{prop:tilde_M_Y_commute}
The observable $\widetilde{M}\left( \tau \right) $ commutes with the quadrature observables $\big\{ \widetilde{Y}_{k}\left( t\right) :k=1,\cdots , \, n,0\leq t\leq
T\big\} $.
\end{proposition}

\begin{proof}
We separate into two cases. First, for $t \le \tau$ we use (\ref{eq:tilde_Y_extend}) to write
\begin{eqnarray}
\widetilde{Y}_{k}\left( t\right) =\widetilde{U}\left( \tau ,0\right) ^{\ast }\left\{
I\otimes Z_{k}\left( t\right) \otimes I\right\} \widetilde{U}\left( \tau ,0\right) ,
\end{eqnarray}
which obviously commutes with $\widetilde{M}\left( \tau \right) $ as given by (\ref{eq:tilde_M}).

Next, for $t\geq \tau $, let us first note that $\widetilde{U}\left( t,\tau \right) 
= U(t , \tau ) \otimes I$ couples the system to the bath processes
over the interval $[\tau ,t]$ and so commutes with $I \otimes I \otimes M $ giving
\begin{eqnarray}
\widetilde{M}\left( \tau \right)  &=&\widetilde{U}\left( \tau ,0\right) ^{\ast }%
\widetilde{U}\left( t,\tau \right) ^{\ast }\left[ I\otimes I\otimes M\right] 
\widetilde{U}\left( t,\tau \right) \widetilde{U}\left( \tau ,0\right)  \nonumber \\
&=&\widetilde{U}\left( t,0\right) ^{\ast }\left[ I\otimes I\otimes M\right] 
\widetilde{U}\left( t,0\right) .
\label{eq:tilde_M_extend}
\end{eqnarray}
We then see that $\left[ \widetilde{M}\left( \tau \right) ,\widetilde{Y}_{j}\left( t\right) \right] $ equals
\begin{eqnarray*}
\widetilde{U}\left( t,0\right) ^{\ast }\left[ I\otimes I\otimes M,I\otimes
Z_{k}\left( t\right) \otimes I\right] \widetilde{U}\left( t,0\right) 
= 0.
\end{eqnarray*}
\end{proof}

\bigskip 

We therefore have that \textit{all} the measurements are of compatible observables.
The question then is what we can now hope to estimate with them.

\begin{proposition}
\label{prop:tilde_jmath_M_Y_commute}
The observable $\widetilde{\jmath}_{t}\left( X\right) $ commutes with
quadratures $\left\{ \widetilde{Y}_{k}\left( u\right) :k=1,\cdots ,n,0\leq u\leq
t\right\} $ and the observable $\widetilde{M}\left( \tau \right) $ provided $%
t\geq \tau $.
\end{proposition}

\begin{proof}
The commutativity of $\widetilde{\jmath}_{t}\left( X\right) $ with the
quadratures up to and including time $t$ follows from an argument virtually
identical to that in Proposition \ref{prop:NDP_basic}. Next, for $t\geq \tau $, we use (\ref{eq:tilde_M_extend})
to write $\left[ \widetilde{\jmath}_{t}\left( X\right) ,\widetilde{M}\left( \tau \right) \right] $ as
\begin{eqnarray*}
\widetilde{U}\left( t,0\right) ^{\ast }\left[ X\otimes I\otimes
I,I\otimes I\otimes M\right] \widetilde{U}\left( t,0\right)  =0.
\end{eqnarray*}
\end{proof}

\bigskip 

Outside of this, there are no natural constraints forcing any of the various
observables under consideration to commute.

\subsection{Summary}
The continuously monitored quadratures $\big\{ \widetilde{Y}_k (t) : k=1, \cdots , n , \, t \in [t_1,t_2] \big\}$ generate a commutative (von Neumann) algebra for all $t_1 < t_2$ which we denote as $\widetilde{\mathfrak{Y}}_{[t_1,t_2]}$. This is the \textit{algebra of the quadrature observations} over the time interval $[ t_1 , t_2 ]$.We have the isotonic condition
\begin{eqnarray}
\widetilde{\mathfrak{Y}}_{[t_1,t_2]} \subset \widetilde{\mathfrak{Y}}_{[t_3,t_4]}
\label{eq:isotony}
\end{eqnarray}
whenever the interval $[t_1, t_2]$ is contained inside $[t_3, t_4]$.
The isotony condition implies that $ \big\{  \widetilde{\mathfrak{Y}}_{[0,t] } : 0 \le t \le T \big\}$ is a nested family of commutative algebras, known as a \textit{filtration}.

The algebra generated by $ \widetilde{\mathfrak{Y}}_{[0,T]}$ and the additional observable $\widetilde{M} (\tau )$ is again a commutative algebra and this is what we intend to condition into.

To find the best estimate of an observable $X$ of the system at time $t \in [0,T]$, i.e., estimate $\widetilde{\jmath} _t (X)$, we use the observations of the quadratures up to and including time $t$ and additionally $\widetilde{M} (\tau )$ if $t \ge \tau$.

The best estimate will have the form $\widetilde{\pi}_t (X) = \text{tr} \{ \widetilde{\rho}_t X \}$ where
\begin{eqnarray}
\widetilde{\rho}_t =
\hat \rho_t , \qquad ( t<\tau )
\end{eqnarray}
where $\hat \rho_t$ is the solution to the basic filter SPDE (\ref{eq:filter_basic}) with initial condition $\rho_0$ at time $0$;
\begin{eqnarray}
\widetilde{\rho}_\tau = \frac{1}{ \text{tr} \big\{ \Phi_m (\hat \rho_\tau )\big\} }\, \Phi_m ( \hat \rho_\tau  ) ; 
\label{eq:widetilde_rho_tau}
\end{eqnarray}
and, for times $t \ge \tau$, $\widetilde{\rho}_t$ will be the solution to the basic filter SPDE (\ref{eq:filter_basic}) with initial condition $\widetilde{\rho}_\tau$ as computed in (\ref{eq:widetilde_rho_tau}) now initialized at time $\tau$.

\begin{remark}
At this stage it is fairly obvious that we may extend the theory to multiple intervention measures at times $ \tau_1 < \tau_2 < \cdots < \tau_r$ during the time interval $[0,T]$. The procedure is simple enough: we use fresh probe systems at each intervention to ensure that all measurement are compatible; we use the basic filter SPDE (\ref{eq:filter_basic}) to propagate the filter starting with state $\rho_0$ at time 0, then update the state by 
\begin{eqnarray}
\widetilde{\rho}_{ \tau_k^-} \mapsto
\widetilde{\rho}_{ \tau_k^+}  =\frac{1}{ \text{tr} \big\{ \Phi^{(k)}_{(m_k)} ( \widetilde{\rho}_{ \tau_k^-} )\big\} }\, \Phi^{(k)}_{(m_k)}(\widetilde{\rho}_{ \tau_k^-}   ) 
\end{eqnarray}
at each time $\tau_k$ where $m_k$ is the recorded value at the $k$ measurement (described by the CP maps $\Phi^{(k)}_{(m_k)}$).
\end{remark}

\subsection{Estimating what the probe measures}
In their paper, Gammelmark et al. \cite{Gammelmark13} consider the situation where the result of the probe measurement at time $\tau$ is not used - locked in a safe deposit box until after the quadrature measurement period $[0,T]$ is over. They then ask as for the probability that the intervention measurement was value $m$ conditional on the whole history of the monitored quadratures.

In our language, we would say that they are estimating the observable $\widetilde{M} (\tau )$ using the quadrature measurements, or equivalently conditioning the observable $\widetilde{M} (\tau )$ onto $ \widetilde{\mathfrak{Y}}_{[0,T]}$.

By Proposition \ref{prop:tilde_M_Y_commute}, this is possible! Let us take the spectral decomposition $M = \sum_m m P_m$, then the goal is to compute
\begin{eqnarray}
p (m , \tau) =
\mathbb{E} [ \tilde{U} (\tau , 0 ) ^\ast [ I \otimes I \otimes   P_m ] \tilde{U} (\tau ,0) \, | \tilde{\mathfrak{Y}}_{[0,T]} ] .
\end{eqnarray}
This may be rewritten as
\begin{eqnarray}
\tilde{U} (\tau , 0 ) \mathbb{E}_{[0,T]} [  [ I \otimes I \otimes   P_m ]  \, | I \otimes \tilde{\mathfrak{Z}}_{[0,T]} 
\otimes I] \tilde{U} (\tau ,0),
\end{eqnarray}
where $\mathbb{E}_{[0,T]}[A]$ is the expectation $ \mathbb{E} [ \tilde{U} (\tau , 0 )^\ast  A \tilde{U} (\tau ,0) ]$.

Now if we fix arbitrary states $| \psi \rangle $ and $\vert \varphi \rangle$ for the probe we have that
\begin{eqnarray*}
\tilde{U} (T,0) \, | \psi \rangle \otimes \vert \Omega \rangle \otimes | \varphi \rangle &=&
\tilde{U} (T,\tau) V U ( \tau , 0 ) \, | \psi \rangle \otimes \vert \Omega \rangle \otimes | \varphi \rangle \nonumber \\
&=& \tilde{F } (T,\tau) V \tilde{F} ( \tau , 0 ) \, | \psi \rangle \otimes \vert \Omega \rangle \otimes | \varphi \rangle 
\end{eqnarray*}
where
$ \tilde{F} (t_2 , t_1 ) = F (t_2 , t_1 ) \otimes I$; compare (\ref{eq:F_Z}).

Using the Theorem, we may write
\begin{eqnarray}
p (m , \tau ) &= & \frac{ \text{tr} \bigg\{ \rho_{\text{sys}}    \otimes \rho_{\text{probe}} \,
[F^Y ( \tau , 0 )^\ast \otimes I]   V^\ast [F^Y ( T , \tau ) \otimes I ] [I \otimes P_m ] 
[ F^Y (T, \tau ) \otimes I ] V [F^Y ( \tau , 0 ) \otimes I] \bigg\} }
{ \text{tr} \bigg\{ \rho_{\text{sys}} \otimes \rho_{\text{probe}} \,
[F^Y ( \tau , 0 )^\ast \otimes I] V^\ast [F^Y ( T , \tau )\otimes I] [I \otimes I ] 
[F^Y (T, \tau ) \otimes I] V   \bigg\} } \nonumber \\
&=& \frac{ \text{tr} \bigg\{ \hat{\rho}_{\text{sys}} (\tau )\otimes \rho_{\text{probe}} \,
 V^\ast [E^Y ( T , \tau ) \otimes   P_m ] 
  V   \bigg\} }
{ \text{tr} \bigg\{  \hat{\rho}_{\text{sys}}  (\tau ) \otimes \rho_{\text{probe}} \,
  V^\ast [E^Y ( T , \tau )\otimes I]   V \bigg\} },
\label{eq:main}
\end{eqnarray}

where $F^Y(t_2 , t_1 ) = \vec{T} e^{\int_{t_1}^{t_2} [ \sum_k L_k dY_k (t) + K dt ] } $,
as in (\ref{eq:F_Y}), $ \hat{\rho}_{\text{sys}}  (\tau )$ is the solution of the filter equation (\ref{eq:filter_basic}) \textit{without} the probe,
and 
\begin{eqnarray}
E^Y(t_2 , t_1 ) = F^Y(t_2 , t_1 )^\ast F^Y(t_2 , t_1 ) .
\end{eqnarray}

\begin{remark}
It is possible - and computationally preferable - to replace $\hat{\rho}_t$ in (\ref{eq:main}) with its unnormalized version $\rho_t$. Indeed, this is what is presented in \cite{Gammelmark13}.
\end{remark}

It is possible to use a reversed Markov process description here to compute $E^Y(T, t )$ as a function of the earliest time $t$. Let us introduce the notation of a backward (past-pointing) It\={o} increment
\begin{eqnarray}  
\overset{\leftarrow}{d} X (t) = X( t ) - X (t -dt ) ,
\end{eqnarray}
where $ dt >0$, then 
\begin{eqnarray*}  
\overset{\leftarrow}{d} F ^Y(T,t) = F^Y ( T , t ) \, \bigg( \sum_k L_k \overset{\leftarrow}{d} Y_k (t) +K dt \bigg) ,
\end{eqnarray*}
from which we find using the It\={o} calculus
\begin{eqnarray}  
\overset{\leftarrow}{d} E ^Y(T,t)&=& \mathcal{L} \bigg( E^Y ( T , t ) \bigg) \, dt \nonumber \\
& +&
 \sum_k\bigg(  E^Y ( T , t ) L_k +L^\ast_k E^Y ( T , t ) \bigg) \, \overset{\leftarrow}{d} Y_k (t).\nonumber \\
\label{eq:E_back}
\end{eqnarray}

Note that the equation (\ref{eq:E_back}) is structurally identical to the Belavkin-Zakai equation (\ref{eq:BZ}) except for the fact that the former now involves backwards It\={o} increments.

In their treatment, Gammelmark et al. take the pair $\big( \rho_\tau , E^Y (T, \tau ) \big)$ to constitute \emph{the} conditioned state of the system (\textit{and} probe) at intermediate time $\tau$ given the measured quadratures over the time $[0,T]$. Their expression has the following alluring feature: $\rho_\tau$ is the unnormalized state which depends on the past measurements $\mathfrak{Y}_{[0, \tau ]}$ and satisfies the Belavkin-Zakai equation (\ref{eq:BZ}) in $\tau$ with initial condition $\rho_0$ while $ E^Y (T , \tau )$ is an effect which depends on the future measurements $\mathfrak{Y}_{( \tau , T ]}$ and satisfies a time reversed Belavkin-Zakai (\ref{eq:E_back}) in $\tau $ with terminal condition $E^Y (T,T ) =I$.

\section{Multitime Interventions}
\label{sec:MultiTime}
The generalization to several intervention measurements is fairly straightforward at this stage. We consider $r$ measurements all made by independent probes: the total probe space will then be the tensor product of the individual probe spaces. The measurements will take place at times $\tau_1 < \tau_2 < \cdots \tau_r $ in the time interval $[0,T]$ and for the $k$th probe we will measure an observable
\[
M_k = \sum_m m P_{M_k} (m) .
\]
For convenience, we will drop the tensor product symbols. The various projections $P_{M_k}$ will then by assumption commute with each other, and with the system and bath observables.

Furthermore we will assume that the entanglement of the system and the $k$th probe immediately before time $\tau_k$ will be implemented by a unitary $V_k$. We seek the probability $p (m_1 , \tau_1 ;  \cdots ; m_r , \tau_r ) $ for a given sequence $(m_1 , \cdots , m_r )$ of measurements conditional on the quadrature measurements over time $[0,T]$.

Let us introduce the mapping
\begin{eqnarray}
\mathscr{G}_k [A] = V^\ast_k [\tilde{F}^Y (t_{k+1} , t_k )^\ast  \otimes I] A \, [\tilde{F}^Y (t_{k+1} , t_k ) \otimes I] V_k \nonumber \\
\label{eq:mathscrG}
\end{eqnarray}
for $k=0, \cdots , r$, where we understand that $t_0 =0, t_{r+1} =T$ and $V_0 = I$.

At this stage, we should add that the conditioning results in all bath operators appearing being diagonal in terms of the measured quadratures. As such we can adopt the view that the family of commuting observables $Y_k (t)$ be simply considered as classical stochastic processes. As a result, the $\mathscr{G}_k$ defined in (\ref{eq:mathscrG}) are super-operators on the system-probe observables, while the $V_k$ and the $I$ in (\ref{eq:main_multitime}) may be understood as acting on the system-probe Hilbert space.

\begin{eqnarray}
p (m_1 , \tau_1 ;  \cdots ; m_r , \tau_r ) 
&= & \frac{ \text{tr} \bigg\{ \rho_{\text{sys}}    \otimes \rho_{\text{probe}} \, \, 
    \mathscr{G}_0  \bigg(      \mathscr{G}_1 	 \bigg( \cdots \mathscr{G}_{r-1} \big(\mathscr{G}_r (P_r) P_{r-1} \big) \cdots P_1 \bigg) \bigg)  \bigg\} }
{  \text{tr} \bigg\{ \rho_{\text{sys}}    \otimes \rho_{\text{probe}} \, \, 
    \mathscr{G}_0  \bigg(      \mathscr{G}_1 	 \bigg( \cdots \mathscr{G}_{r-1} \big(\mathscr{G}_r (I) I \big) \cdots I \bigg) \bigg)  \bigg\} } ,
\label{eq:main_multitime}
\end{eqnarray}

\section{Conclusion}
To set things in historical context, the mid-1920's saw the development of quantum theory as a fundamentally new theory of Nature. The Born statistical interpretation introduced the probabilistic feature to the theory, while Heisenberg's noncommutativity of observables and uncertainty principle signaled that this would be a major departure from classical probability theory. This culminated in von Neumann's axiomatic formulation of quantum mechanics. Actually, the axioms of classical probability by Kolmogorov appeared only in the following year, but rapidly led to widespread applications in science and engineering. The estimation problem for classical systems was pioneered by control theorists such as Stratonovich and Kalman with later contributions by Zakai, Duncan, etc., and is now a mature field in communications and signal processing. 

In the classical case, the input-output models are causal and there is no issue with incompatibility of random variables: therefore estimation of future, present and past variables is not a problem. The situation changes drastically in quantum theory. The formulation of quantum filtering was given by Belavkin and showed how the results of Stratonovich and Kalman could be extended to the quantum domain. 

The question of estimating the past is problematic in quantum theory. It does not make sense to estimate which slit a quantum particle went through in a two slit experiment if you have not measured it.
We know Bayes Theorem does not work, so a blind use of the conditioning rules is unwarranted. The question as to when we may legitimately use (as opposed to misuse) estimation theory for quantum retrodiction is one that needs a detailed analysis. But the questions posed in the physics literature are of interest as the push to concepts of quantum measurement and estimation to new areas, and add new layers to the traditional notions of quantum states and observables. Fortunately, the simper intuition comes out intact for the class problems considered but this required a subtle analysis of what is actually going on.

In the paper, we show the consistency of the condition rule for (multiple) external measurements on an open quantum system during a period where it is undergoing simultaneaous monitoring.

\bigskip

\textbf{Acknowledgement} The author is grateful to Klaus M\o lmer and Benjamin Huard for comments.

\end{document}